\UseRawInputEncoding
\documentclass[submission,copyright,creativecommons]{eptcs}
\providecommand{\event}{LSFA 2021} 
\usepackage{underscore}           

\hypersetup{
  colorlinks=true,
  linkcolor=blue,
  anchorcolor=blue,
  citecolor=blue,
  urlcolor=black,
  filecolor=blue,
  breaklinks=true
}

\usepackage[utf8]{inputenc}
\usepackage{nicefrac}
\usepackage{amsmath}
\usepackage{amsthm}
\usepackage{amsfonts}
\usepackage{proof}
\usepackage{wasysym}
\usepackage[T1]{fontenc}
\usepackage[framestyle=fbox,framefit=yes,heightadjust=all,framearound=all]{floatrow}    

\usepackage{tikz-cd}

\usetikzlibrary{decorations.pathmorphing}

\newcommand\coin{\fullmoon}
\newcommand\lambprot{\ensuremath{\lambda_{\circ}}}

\newcommand\topr[1][1]{\ensuremath{\rightarrow_{#1}}}
\newcommand\ite[3]{\ensuremath{\mathsf{if~}#1\mathsf{~then~}#2\mathsf{~else~}#3}}
\newcommand{\FV}[1]{\ensuremath{\mathsf{FV}(#1)}}

\newtheorem{thm}{Theorem}[section]
\newtheorem{lem}[thm]{Lemma}

\theoremstyle{definition}
\newtheorem{defin}[thm]{Definition}

\title{A Note on Confluence\\ in Typed Probabilistic Lambda Calculi\thanks{This paper is part of the first author's masters thesis~\cite{Romero20}. Founded by PICT-2019-1272, STIC-AmSud 21STIC10 Qapla', ECOS-Sud A17C03 QuCa, and PIP 11220200100368CO.}}
\author{Rafael Romero
  \institute{
    CONICET-Universidad de Buenos Aires.\\
    Instituto de Ciencias de la Computaci\'on.\\
  Buenos Aires, Argentina.}
  \email{lromero@icc.fcen.uba.ar}
  \and
  Alejandro D\'iaz-Caro
  \institute{Departamento de Ciencia y Tecnolog\'ia.\\
    Universidad Nacional de Quilmes.\\
  Bernal, Buenos Aires, Argentina.}
  \institute{CONICET-Universidad de Buenos Aires.\\ 
    Instituto de Ciencias de la Computaci\'on.\\
  Buenos Aires, Argentina.}
  \email{adiazcaro@icc.fcen.fcen.uba.ar}
}
\def\titlerunning{A Note on Confluence in Typed Probabilistic Lambda Calculi}
\def\authorrunning{R. Romero \& A. D\'iaz-Caro}
\begin{document}
\maketitle

\begin{abstract}
On the topic of probabilistic rewriting, there are several works studying both termination and confluence of different systems. While working with a lambda calculus modelling quantum computation, we found a system with probabilistic rewriting rules and strongly normalizing terms. We examine the effect of small modifications in probabilistic rewriting, affine variables, and strategies on the overall confluence in this strongly normalizing probabilistic calculus.
\end{abstract}

\section{Introduction}

When dealing with probabilistic lambda calculus, we can find two different sources of divergence.
\begin{itemize}
\item A \textit{single redex} may reduce in two different ways via a probabilistic reduction.

\item A term with \textit{multiple redexes} and no strategy, could be reduced in different ways.
\end{itemize}
For example, we can consider a lambda calculus extended with a coin $\coin$ reducing to $0$ or $1$ with probability $\frac 12$ each. Then, taking just the coin, we are in the first case of divergence. While taking, for example, $(\lambda x.\lambda y.yxx)\coin$, we are in the second case, since we can either beta reduce, or reduce the coin.

There is no point in trying to achieve confluence in the first case: the coin is non-confluent by design. However, we can analyse the branching paths and verify that the probability of reducing to a particular term stays the same, regardless of the reduction sequence. This is what we call \textit{probabilistic confluence}.

To study this kind of cases, Bournez and Kirchner developed the notion of PARS~\cite{BK02}, later refined in \cite{BG06}. Using the techniques described in~\cite{DCM17,Faggian} we can define the rewriting rules over the distributions. 

If we denote by $[(p_1,t_1),\dots,(p_n,t_n)]$ the probability distribution where $t_i$ has probability $p_i$, the possible reductions from the previous example are depicted in Figure~\ref{fig:noconfl}. The resulting distributions are not only different, but also divergent, since in the left branch, the probability to arrive, for example, to $\lambda y.y01$ is $0$, while it is one of the possible results in the right branch.

\begin{figure}[t]
  \centering
  \begin{tikzcd}[column sep=small]
    {[({\frac 12}, (\lambda x.\lambda y.yxx)0), ({\frac 12}, (\lambda x.\lambda y.yxx)1)]}\ar[d] & (\lambda x.\lambda y.yxx)\coin \ar[l]\ar[r]&  {[(1, \lambda y.y\coin\coin)]}\ar[d,"*",pos=1]\\
    {[({\frac 12}, \lambda y.y00), ({\frac 12}, \lambda y.y11)]} 
    &&\hspace{-50pt}
    {[({\frac 14}, \lambda y.y00),({\frac 14}, \lambda y.y01),({\frac 14}, \lambda y.y10),({\frac 14}, \lambda y.y11)]} 
  \end{tikzcd}
  \caption{Counterexample of confluence}
  \label{fig:noconfl}
\end{figure}

In this short paper we consider a simply typed lambda calculus extended with 
a coin, and show different possibilities for achieving some sort of confluence, without giving preference to any of them.

In Section~\ref{sec:calculus} we introduce the calculus to be studied, without any restrictions either in the rewriting rules, or in the typing rules. As we argued above, this naive definition is not confluent (cf.~Figure~\ref{fig:noconfl}), unless a strategy is defined (in which case it becomes trivially probabilistically confluent, as will be discussed in Section~\ref{sec:strategy}).
In Section~\ref{sec:circulito} we show that we can achieve confluence by internalising the probabilistic reductions in the terms. In Section~\ref{sec:affine} we show that we can achieve probabilistic confluence by taking an affine-linear type system. Then, in Section~\ref{sec:subaffine}, we show that we can relax the type system in an if-then-else branching, obtaining a probabilistic confluence result modulo a computational equivalence.

\section{The \texorpdfstring{$\lambprot$}{lambda coin} calculus}~\label{sec:calculus}
In this section we present 
$\lambprot$ (read ``lambda coin''), which is the
simply typed lambda calculus extended with booleans ($1$ and $0$), an if-then-else construction, and a coin $\coin$.
Terms are inductively defined by 
\[
  t := x\mid \lambda x.t \mid tt \mid 1\mid 0\mid \ite ttt\mid \coin
\]

The rewrite system is given in Table~\ref{tab:TRS}. The rules $t\topr[p]r$ mean that $t$ reduces with probability $p$ in one step to $r$, where the sum of probabilities of reducing one redex is $1$. In particular, every non-contextual rule has probability $1$, since there is only one rule per redex, except for the coin, which reduces with probability $\frac 12$ to $0$ and probability $\frac 12$ to $1$. If $t\topr[p_1]r_1\dots\topr[p_n]r_n$ we may write $t\topr[\prod_{i=1}^np_i]^*r_n$.

The reduction rules are intentionally as permissive as possible (even the branches of an if-then-else can be reduced) in order to analyse its (lack of) confluence. The type system is given in Table~\ref{tab:TS}.

\begin{table}[t]
  \[
    \begin{array}{r@{\ }l@{\qquad}r@{\ }l@{\qquad}r@{\ }l}
      (\lambda x.t)r &\topr t[r/x]
      &
      \ite 1tu &\topr t
      &
      \coin &\topr[{\frac 12}] 1
      \\
      &&
      \ite 0tu &\topr u
      &
      \coin &\topr[{\frac 12}] 0
    \end{array}
  \]
  \smallskip
  \[
    \infer{\lambda x.t\topr[p]\lambda x.r}{t\topr[p] r}
    \qquad
    \infer{ts\topr[p] rs}{t\topr[p] r}
    \qquad
    \infer{st\topr[p] sr}{t\topr[p] r}
    \qquad
    \infer{\ite t u v\topr[p]\ite r u v}{t\topr[p] r}
  \]
  \[
    \infer{\ite t u v\topr[p]\ite t r v}{u\topr[p] r}
    \qquad
    \infer{\ite t u v\topr[p]\ite t u r}{v\topr[p] r}
  \]
  \caption{Rewrite system for $\lambprot$.}
  \label{tab:TRS}
\end{table}

\begin{table}[t]
  \[
    A:= \mathbb{B}  \mid A\rightarrow A
  \]
  \[
    \infer[\mathsf{ax}]{\Gamma,x:A\vdash x:A}{} 
    \qquad
    \infer[\rightarrow_i]{\Gamma\vdash\lambda x.t:A\rightarrow B}
    	{\Gamma,x:A\vdash t:B}
	\qquad 
    \infer[\rightarrow_e]{\Gamma\vdash tr:B}
    	{\Gamma\vdash t:A\rightarrow B & \Gamma\vdash r:A}
  \]
  \[
    \infer[\mathsf{ax}_1]{\Gamma\vdash 1:\mathbb{B}}{}
    \qquad
    \infer[\mathsf{ax}_0]{\Gamma\vdash 0:\mathbb{B}}{}
    \qquad
    \infer[\mathsf{if}]{\Gamma\vdash\ite t u v:A}{\Gamma\vdash t:\mathbb{B}& \Gamma\vdash u:A & \Gamma\vdash v:A}
    \qquad
    \infer[\mathsf{ax}_{\coin}]{\Gamma\vdash \coin:\mathbb{B}}{} 
  \]
  \caption{Type system for $\lambprot$.}
  \label{tab:TS}
\end{table}

Strong normalization for this calculus follows trivially from Tait's proof for the simply typed $\lambda$-calculus, extended with booleans. The only reduction added is the probabilistic coin toss and it takes at most one step for each operator. Hence, using the rewriting over probabilistic distributions techniques from~\cite{DCM17}, we only need to show local confluence in order to achieve global confluence altogether. This is an adaptation of Newman's lemma for probabilistic calculi (see~\cite{DCM17} for a longer discussion about probabilistic confluence).

Clearly, $\lambprot$ is not confluent, as already seen in the introduction (see Figure~\ref{fig:noconfl}). It is easy to see that these distributions represent different results.

The divergence stems from three characteristics of the calculus:
  (1) Lack of a reduction strategy.
  (2) Probabilistic reductions.
  (3) Duplication of variables.

Removing just one of these elements renders the system confluent, however each modification comes with its own trade-off. We will examine each case, one by one, in the following section.

\section{Removing the divergence sources}
\subsection{Defining a strategy}\label{sec:strategy}
The definition of a strategy is the easiest modification. Choosing a reduction strategy makes all critical pairs disappear, since there is only one possible reduction rule to be applied for each term distribution. For example, in Figure~\ref{fig:noconfl} a \textit{call-by-name} strategy would take the right path, where \textit{call-by-value} would take the left one. 

Ultimately the choice lies in how to interpret the duplication of variables. Reducing via call-by-name means that the probabilistic event is duplicated. Whereas, a call-by-value strategy duplicates the outcome of said event (see~\cite{Dallago} for a discussion on this choice, and even an alternative combining both).

\subsection{Internalising the probabilities}\label{sec:circulito}
Following~\cite{lambdarho}, we can modify the reduction on $\lambprot$ to internalise the entire distribution of a term. In this particular case, every reduction has probability $1$, and the coin toss deterministically reduces to its probability distribution.
We write $t\oplus_{p}r$ for the probability distribution $[(p,t),(1-p,r)]$, $(t\oplus_p r)\oplus_q s$ for the probability distribution $[(qp,t),(q(1-p)r),(1-q,s)]$, etc. Then, we can consider the rewrite rule
\(
  \coin\to 0 \oplus_{\frac 12} 1
\).
This idea is common in non-probabilistic settings as well, e.g.,~\cite{AlvesDunduaFloridoKutsiaIGPL18}.

Following this approach brings confluence to the calculus, since every repetition of a probabilistic event rewrites to the same result, its distribution. For example, the two branches of Figure~\ref{fig:noconfl} become $(\lambda x.\lambda y.yxx)(0\oplus_{\frac 12}1)$ and $\lambda y.y\coin\coin$, both converging to $\lambda y.y(0\oplus_{\frac 12}1)(0\oplus_{\frac 12}1)$.

Although this is a valid solution, it forces us to consider every possible state of a program at the same time along with its probability of occurrence, making the management of the system more complex. Here we are dealing with a simple coin, but more involved calculi might have several different reductions, each with its own distribution. If not designed correctly, a language that holds every possible state in the probability distribution can easily become too cumbersome to be effective.

\subsection{Affine variables}\label{sec:affine}
The last reasonable solution is to restrict duplication. One way to do this is by controlling the appearance of variables at the type system level, with an affine type system, see Table~\ref{tab:AffineTS}.

\begin{table}[t]
  \[
    A:= \mathbb{B}  \mid A\rightarrow A
  \]
  \[
    \infer[\mathsf{ax}]{\Gamma,x:A\vdash x:A}{} 
    \qquad
    \infer[\rightarrow_i]{\Gamma\vdash\lambda x.t:A\rightarrow B}
    	{\Gamma,x:A\vdash t:B}
	\qquad 
    \infer[\rightarrow_e]{\Gamma,\Delta\vdash tr:B}
    	{\Gamma\vdash t:A\rightarrow B & \Delta\vdash r:A}
  \]
  \[
    \infer[\mathsf{ax}_1]{\Gamma\vdash 1:\mathbb{B}}{}
    \qquad
    \infer[\mathsf{ax}_0]{\Gamma\vdash 0:\mathbb{B}}{}
    \qquad
    \infer[\mathsf{if}]{\Gamma,\Delta_1,\Delta_2\vdash\ite t u v:A}{\Gamma\vdash t:\mathbb{B}& \Delta_1\vdash u:A & \Delta_2\vdash v:A}
    \qquad
    \infer[\mathsf{ax}_{\coin}]{\Gamma\vdash \coin:\mathbb{B}}{} 
  \]
  \caption{Affine type system for $\lambprot$
  (different contexts are considered to be disjoint).
  }
  \label{tab:AffineTS}
\end{table}

This type system solves the counterexample from Figure~\ref{fig:noconfl}, since the considered term has no type in this system. In particular, we can prove the following property:
\begin{lem}
  If $r\topr[p] s$ then $t[r/x]\topr[p] t[s/x]$.
\end{lem}
Notice that this property is not true in the unrestricted $\lambprot$. For example, while $\coin\topr[\frac 12]1$, we have $(\lambda y.yxx)[\coin/x]=\lambda y.y\coin\coin\topr[\frac 14]^*\lambda y.y11$.

The drawback in this approach is clear, there is a loss in expressivity. 
Of course, in some cases, this restriction is a desirable quality. For example, in quantum computing it may serve 
to avoid cloning qubits, a forbidden operation in quantum mechanics.

\section{Computational confluence with a sub-affine type system}\label{sec:subaffine}
The solution considered in Section~\ref{sec:affine} seems quite extreme. In particular, using the same variable in different branches of an if-then-else construction does not actually duplicate it, since only one of those branches will remain. However, changing the rule $\mathsf{if}$ from Table~\ref{tab:AffineTS} to $\mathsf{if}_s$ given by
\[
  \infer[\mathsf{if}_s]{\Gamma,\Delta\vdash\ite t u v:A}{
    \Gamma\vdash t:\mathbb{B}
    & 
    \Delta\vdash u:A 
    &
    \Delta\vdash v:A
  }
\]
breaks confluence anyway. We call this calculus ``sub-affine''.
Consider the following example. Let $\mathsf{not}=\lambda x.\ite x01$, then

\begin{center}
  \begin{tikzcd}[row sep=11pt,column sep=-50pt]
  & (\lambda x.\lambda y. \ite y {x}{\mathsf{not}~x})\coin \ar[dl]\ar[dr]& \\
 \left[
 \begin{tabular}{c}
   $(1/2, (\lambda x.\lambda y. \ite y {x}{\mathsf{not}~x})0);$\\
   $(1/2, (\lambda x.\lambda y. \ite y {x}{\mathsf{not}~x})1)$
 \end{tabular} 	
 \right] \ar[d]& &
 \left[(1, (\lambda y. \ite y {\coin}{\mathsf{not}~\coin}))\right]\ar[d]\\
 \left[
 \begin{array}[c]{l}
 	(\nicefrac 12, \lambda y. \ite y 01);\\
	(\nicefrac 12, \lambda y. \ite y 10) 
 \end{array}
 \right] & &
 \left[	
   \begin{array}[c]{l}
  	(\nicefrac 14, \lambda y. \ite y 00);\\
 	(\nicefrac 14, \lambda y. \ite y 01);\\
 	(\nicefrac 14, \lambda y. \ite y 10);\\
 	(\nicefrac 14, \lambda y. \ite y 11)
   \end{array}
 \right]
\end{tikzcd}
\end{center}

Note that terms in both distributions are in normal form. The two paths are syntactically divergent, however the resulting programs share the same behaviour under the same inputs. If we were to apply the resulting abstractions to $0$ or to $1$, both paths would yield $[(\nicefrac 12, 0); (\nicefrac 12, 1)]$. Therefore, these distributions are semantically confluent, they are not the same terms but they \textit{represent} the same function.

We can formalise this notion for the sub-affine calculus as follows. Let 
\(
  C := \lozenge\mid Cv
\)
be an elimination context, where $\lozenge$ is called ``placeholder'' and $v$ is a normal closed term. We write $C\langle t\rangle$ for $C[t/\lozenge]$. Notice that $C\langle t\rangle$ is a term. We say that $C$ is an elimination context of $A$, written $C^A$, if for all $\vdash t:A$, we have $\vdash C\langle t\rangle:\mathbb B$. That is, it applies $t$ until it reaches the basic type $\mathbb B$.

The computational equivalence is defined as follows.
\begin{defin}
Let $D_1 = [(p_i,t_i)]_i$ and $D_2 = [(q_j,r_j)]_j$ be two distributions of
terms, all closed of type $A$. Then, we say that these distributions are
computationally equivalent (notation $D_1\equiv D_2$) if for all $C^A$ we have
$C\langle t_i\rangle\topr[u_{ik}]^* b_{ik}$ and $C\langle r_j\rangle\topr[s_{jh}]^* c_{jh}$ such that the
$b_{ik}$ and $c_{jh}$ are in normal form, and  $[(p_iu_{ik},b_{ik})]_{ik} \sim
[(q_js_{jh},c_{jh})]_{jh}$
, where $\sim$ denotes the equality on distributions.
\end{defin}
The previous definition means that two distributions are computationally equivalent if by applying the resulting terms to all the possible inputs, they produce the same probability distribution of results. Notice that the definition is not assuming confluence. 

Then, we can prove that the confluence modulo computational equivalence of the sub-affine calculus (Theorem~\ref{thm:confmodulo}). We need the following two lemmas, which follow by straightforward induction.
\begin{lem}\label{lem:subst}
Let $y\not\in\FV{r}$, then $t[q/y][r/x]= t[r/x][q[r/x]/y]$.
\qed
\end{lem}

\begin{lem}\label{lem:distRed}
Let $t$ reduce to the distribution $D$, then $t[r/x]$ reduces to $D[r/x]$.
\qed
\end{lem}

\begin{thm}[Computational confluence]\label{thm:confmodulo}
  Let $\vdash t:A$ in the sub-affine calculus.
  If $t$ reduces to the distribution $D_1$ and to the distribution $D_2$, then $D_1\equiv D_2$.
\end{thm}
\begin{proof}
  We only consider the six critical pairs, 
  four derived from the if-then-else construction and two from the classical lambda calculus.
  Non-critical pairs are trivially probabilistically confluent.
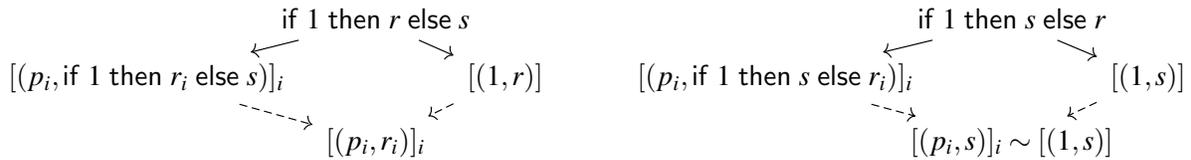
\begin{center}
  \begin{tikzcd}[column sep=-1em,row sep=1ex]
    & \ite{1}{r}{s} \ar[dl]\ar[dr] & \\
    {[(p_i, \ite{1}{r_i}{s})]_i}\ar[dr,dashed]  && {[(1, r)]}\ar[dl,dashed] \\
    & {[(p_i, r_i)]_i}  &
  \end{tikzcd}
  \hfill
  \begin{tikzcd}[column sep=-1em,row sep=1ex]
    & \ite{1}{s}{r} \ar[dl]\ar[dr] & \\
    {[(p_i, \ite{1}{s}{r_i})]_i}\ar[dr,dashed] &&  {[(1, s)]}\ar[dl,dashed]\\
    & {[(p_i, s)]_i} \sim [(1, s)] 
  \end{tikzcd}
\end{center}

The symmetrical cases with $0$ close in a similar way. 
The fifth critical pair closes by Lemma \ref{lem:distRed}.
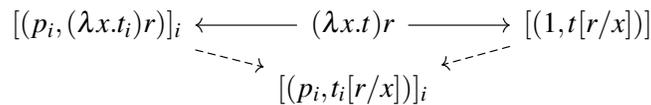
\begin{center}
  \begin{tikzcd}[row sep=1ex]
    {[(p_i, (\lambda x.t_i)r)]_i}\ar[dr,dashed]
    & (\lambda x.t)r\ar[l]\ar[r]& 
    {[(1, t[r/x])]}\ar[dl,dashed]\\
    & {[(p_i, t_i[r/x])]_i} &
  \end{tikzcd}
\end{center}

The last critical pair requires a more thorough analysis. 
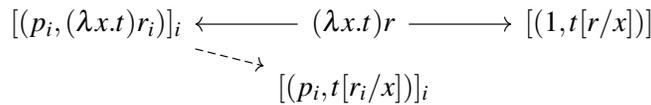
\begin{center}
  \begin{tikzcd}[row sep=1ex]
    {[(p_i, (\lambda x.t)r_i)]_i}\ar[dr,dashed] 
    & (\lambda x.t)r\ar[l]\ar[r]& 
    {[(1, t[r/x])]}\\
    &{[(p_i,t[r_i/x])]_i} & 
  \end{tikzcd}
\end{center}

We must prove that for all $C^A$ and for all $i$, there exists $b_{ij}$ with $C\langle t[r_i/x]\rangle\topr[q_{ij}]^* b_{ij}$ and $c_k$ with $C\langle t[r/x]\rangle\topr[s_{k}]^* c_k$ such that
\(
  [p_iq_{ij},b_{ij}]_{ij} \sim [s_k,c_k]_k
\).
It is enough to take only one path to $[p_iq_{ij},b_{ij}]_{ij}$ and to $[s_k,c_k]_k$ according to the definition of $\equiv$.
If $x\notin\FV{t}$, then 
    \(
      C\langle t[r/x]\rangle
      =
      C\langle t\rangle
      =
      C\langle t[r_i/x]\rangle
    \).
  If $x$ appears once in $t$, 
    \(
      C\langle t[r/x]\rangle
      \topr[p_i] C\langle t[r_i/x]\rangle
    \).
    If $x$ appears more than once in $t$, it appears at most $2^n$ times (where $n$ is the number of if-then-else constructions), so we can proceed by induction on $n$.
    \begin{itemize}
      \item 
	If $n=1$, then there is one if-then-else, say $\ite c{s_1}{s_2}$ and $x$ appears both in $s_1$ and in $s_2$. Therefore, there are three cases:
	\begin{itemize}

	  \item $c$ is $0$ or $1$, then, for $j=0$ or $j=1$, we have
	    \(
	      C\langle t[r/x]\rangle
	      \topr[1]^*
	      C\langle s_j[r/x]\rangle
	      \topr[p_i]^*
	      C\langle s_j[r_i/x]\rangle
	    \).
	    Notice that $C\langle t[r_i/x]\rangle\topr[1] C\langle s_j[r_i/x]\rangle$.
	  \item $c\topr[q_j]^* j$ with $j=0,1$, so we have
	    \(
	      C\langle t[r/x]\rangle
	      \topr[q_j]^*
	      C\langle s_j[r/x]\rangle
	      \topr[p_i]^*
	      C\langle s_j[r_i/x]\rangle
	    \).
	    That is, we arrived to the distribution $[(q_jp_i,s_j[r_i/x])]_{ij}$.
	    Notice that since $C\langle t[r_i/x]\rangle\topr[q_j]
	    C\langle s_j[r_i/x]\rangle$, we arrive to the same distribution in the right branch of this critical pair.
	  \item $c$ is an open term, hence not reducing to a constant $0$ or $1$.
	    In such a case, since $t[r/x]$ is a closed term, it means that the if-then-else construction in $t$ is under a lambda-abstraction, therefore, we must beta-reduce first, either with the argument in $t$, or with an external argument given by the context $C$. We can repeat the same process until we get one of the previously treated cases (that is, at some point, the if-then-else becomes a redex).
	\end{itemize}
      \item If $n>1$, we proceed as the previous case, reducing the if-then-else first, which reduces $n$ and so the induction hypothesis applies.\qedhere
    \end{itemize}
\end{proof}

\section{Conclusion}
In this paper we have analysed the different possibilities to transform a simple probabilistic calculus into a (probabilistically / computationally) confluent calculus. The main contribution on this note is Theorem~\ref{thm:confmodulo}, which proves that we can relax the affinity restriction on ``non-interfering'' paths. This technique has been considered in the first author master's thesis~\cite{Romero20} to prove the computational confluence of the quantum lambda calculus $\lambda_\rho$~\cite{lambdarho}.

\bibliographystyle{eptcs}
\bibliography{biblio}
\end{document}